\documentclass[a4paper,11pt]{article}
\usepackage{graphicx} 
\usepackage{amsfonts,amsmath,amsthm,color}
\usepackage[ansinew]{inputenc}
\usepackage[T1]{fontenc}
\newtheorem{theorem}{Theorem}[section] 
\newtheorem{corollary}[theorem]{Corollary}
\newtheorem{proposition}[theorem]{Proposition}
\newtheorem{definition}[theorem]{Definition}
\newtheorem{example}[theorem]{Example}

\begin{document}
\title{Uniform continuity of POVMs}
\author{Roberto Beneduci\thanks{e-mail rbeneduci@unical.it}\\
{\em Dipartimento di Fisica,}\\
{\em Universit\`a della Calabria,}\\
and\\
{\em Istituto Nazionale di Fisica Nucleare, Gruppo c. Cosenza,}\\
{\em 87036 Arcavacata di Rende (Cs), Italy,}\\
}
\date{}
\maketitle
\begin{abstract}
Recently a characterization of uniformly continuous POVMs and a necessary condition for a uniformly continuous POVM $F$ to have the norm-1 property have been provided. 
 Moreover it was proved that in the commutative case, uniform continuity corresponds to the existence of a Feller Markov kernel. We apply such results to the analysis of some relevant physical examples; i.e.,  the phase space localization observables, the unsharp phase observable and the unsharp number observable of which we study the uniform continuity, the norm-1 property and the existence of a Feller Markov kernel. 


\end{abstract}

\section{Introduction}
\label{intro}
Positive operator valued measures (POVM) are a key mathematical object in the modern formulation of quantum mechanics. They are a natural consequence of the analysis of the statistics of the quantum measurement process and can be used as the mathematical representatives of the quantum observables \cite{Ali3,Busch,Davies,Holevo,Schroeck}. A very interesting example of application of POVMs can be found in quantum information theory, where a measurement is generally described by a POVM.  We recall that the set of projection valued measures (PVM) is a subset of the set of POVMs and that there is a one-to-one correspondence between PVMs and self-adjoint operators. The observables described by POVMs are usually called unsharp observables while the ones described by PVMs are called sharp observables. 

 It is worth remarking that the use of POVMs opens new possible ways of analysis of the quantum realm. For example, the position and momentum observables are jointly measurable when described by POVMs \cite{Davies,Holevo,Ali4,Busch,Prugovecki,Schroeck} and that allows an interpretation of the Heisenberg inequality in terms of the inaccuracy of the joint measurement of position and momentum. Such an interpretation is not compatible with the formalism of standard quantum mechanics where position and momentum are described by self-adjoint operators \cite{Holevo}. 

Moreover, there are observables (e.g. photon localization, phase, time) for which a description in terms of PVMs is ruled out by their very nature.  For example, that the localization of the photon cannot be represented by a PVM is a consequence of the transversality condition on the electromagnetic field \cite{Kraus,Ali3}.

All that raises the problem of the mathematical characterization and physical interpretation of the POVMs. Unfortunately, from the mathematical point of view, such a problem is much more involved than the problem of the characterization of PVMs. Anyway, a sharp characterization is possible in the commutative case.  
Indeed, it was pointed out \cite{Ali,Holevo,Naimark,B1,B2,B3,B4,B5,B6,B7,B8,B9,B10,B11,P} that there is a strong relationship between commutative POVMs and PVMs. In particular, it was recently proved \cite{B1} that a POVM $F$ is commutative if and only if  there exist a spectral measure $E$ and a Feller Markov kernel $\mu_{(\cdot)}(\cdot):\Gamma\times\mathcal{B}(\mathbb{R})\to[0,1]$, $\Gamma\subset\sigma(A)$, $E(\Gamma)=\mathbf{1}$, such that 

\begin{equation}\label{Feller}
 F(\Delta)=\int_{\Gamma}\mu_{\Delta}(\lambda)\,dE_{\lambda}
\end{equation}

\noindent
and $\mu_{\Delta}(\cdot)$ is continuous for each $\Delta\in R$ where, $R\subset\mathcal{B}(\mathbb{R})$ is a ring which generates the Borel $\sigma$-algebra of the reals $\mathcal{B}(\mathbb{R})$. That is interesting not only from the mathematical point of view but also from the physical point of view since there are important observables which are described by commutative POVMs as for example the unsharp position and momentum observables. 

A second mathematical question of physical relevance is connected to the continuity property of the POVMs. Indeed, there are relevant physical examples \cite{B2,B12,B13,Schroeck} of observables described by POVMs that are continuous with respect to the uniform operator topology. Recently, necessary and sufficient conditions for the uniform continuity have been given and the connections with the concepts of absolute continuity and norm-1 property have been analyzed \cite{B2}. Moreover, it was proved \cite{B1} that a commutative POVM $F$ is uniformly continuous if and only if the Markov kernel in equation (\ref{Feller}) can be chosen to be a strong Feller Markov kernel.
         
The aim of the present paper is to describe the results in Ref.s \cite{B1,B2} that we have just outlined and to analyze some physically relevant examples in the light of such results. 


\noindent 



In particular, we analyze 
the phase space localization observables, the unsharp phase observable and the unsharp number observable. We also show \cite{B2} that, in several cases, the uniform continuity of a POVM $F$ allows a direct analysis of the norm-1 property. That is tested on the previous examples.

\section{Definition and main properties of POVMs}
\noindent
In what follows, we denote by $\mathcal{B}(X)$ the Borel $\sigma$-algebra of a topological space $X$. 
By $\textbf{0}$ and $\textbf{1}$ the null and the identity operators, by $\mathcal{L}_s(\mathcal{H})$ the space of all bounded self-adjoint linear operators acting in a Hilbert space $\mathcal{H}$ with scalar product $\langle\cdot,\cdot\rangle$, by $\mathcal{F}(\mathcal{H})=\mathcal{L}_s^+(\mathcal{H})$ the subspace of all positive, bounded self-adjoint operators on $\mathcal{H}$, by $\mathcal{E}(\mathcal{H})\subset\mathcal{F}(\mathcal{H})$ the subspace of all projection operators on $\mathcal{H}$. 

\begin{definition}
\label{POV}
A Positive Operator Valued measure (for short, POVM) is a map $F:\mathcal{B}(X)\to\mathcal{
F}(\mathcal{H})$
such that:
    \begin{equation*}
    F\big(\bigcup_{n=1}^{\infty}\Delta_n\big)=\sum_{n=1}^{\infty}F(\Delta_n).
    \end{equation*}
    \noindent 
 where, $\{\Delta_n\}$ is a countable family of disjoint
    sets in $\mathcal{B}(X)$ and the series converges in the weak operator topology. It is said to be normalized if 
\begin{equation*}   
    F(X)={\bf{1}}
\end{equation*}
\end{definition}    
\begin{definition}
    A POVM is said to be commutative if
    \begin{equation}
    \big[F(\Delta_1),F(\Delta_2)\big]={\bf{0}},\,\,\,\,\forall\,\Delta_1\,,\Delta_2\in\mathcal{B}(X).
    \end{equation}
    \end{definition}

   \begin{definition}
   A POVM is said to be orthogonal if
    \begin{equation}\label{orthogonal}
    F(\Delta_1)F(\Delta_2)={\bf{0}}\,\,\,\hbox{if}\,\,\Delta_1\cap\Delta_2=
    \emptyset.
    \end{equation}
\end{definition}
\begin{definition}\label{PVM}
A Spectral measure or Projection Valued measure (for short, PVM) is an orthogonal, normalized POVM.
\end{definition}
\noindent
Let $E$ be a PVM. By equation (\ref{orthogonal}), 
$$\mathbf{0}=E(\Delta)E(X-\Delta)=E(\Delta)[\mathbf{1}-E(\Delta)]=E(\Delta)-E(\Delta)^2.$$ 
\noindent
We can then restate definition \ref{PVM} as follows.
\begin{definition}
A PVM $E$ is a POVM such that $E(\Delta)$ is a projection operator for each $\Delta\in\mathcal{B}(X)$.
\end{definition}

\noindent
In quantum mechanics, non-orthogonal normalized POVM are also called \textbf{generalised} or \textbf{unsharp} observables and PVM \textbf{standard} or \textbf{sharp} observables. 

\noindent
In what follows, we shall always refer to normalized POVMs and we shall use the term ``measurable'' for the Borel measurable functions.
For any vector $x\in\mathcal{H}$ the map
$$\langle F(\cdot)x,x\rangle \,:\,\mathcal{B}(X)\to {\mathbb R} ,
\qquad
\Delta \mapsto \langle F(\Delta)x,x\rangle,$$
is a Lebesgue-Stieltjes measure. 
In the following, we shall use the symbol $d\langle F_{\lambda}x,x\rangle$ to mean integration with respect to the measure $\langle F(\cdot)x,x\rangle$.
\noindent
A measurable function $f:N\subset X\to f(N)\subset\mathbb{R}$ is said to be almost everywhere (a.e.) one-to-one with respect to a POVM $F$ if it is one-to-one on a subset $N'\subset N$ such that $F(N-N')=\mathbf{0}$.
A function $f:X\to\mathbb{R}$ is bounded with respect to a POVM $F$, if it is equal to a bounded function $g$ a.e. with respect to $F$, that is, if $f=g$ a.e. with respect to the measure $\langle F(\cdot)x,x\rangle$,  $\forall x \in \mathcal{H}$.
\noindent
For any real, bounded and measurable function $f$ and for any POVM $F$, there is a unique \cite{Berberian} bounded self-adjoint operator $B\in\mathcal{L}_s(\mathcal{H})$ such that
\begin{equation}
\label{6}
\langle Bx,x\rangle=\int f(\lambda)d\langle F_{\lambda}x,x\rangle,\quad\text{for each}\quad x\in\mathcal{H}.
\end{equation}
If equation (\ref{6}) is satisfied, we write $B=\int f(\lambda)dF_{\lambda}$ or $B=\int f(\lambda)F(d\lambda)$ equivalently. 
\noindent
\begin{definition}
The spectrum $\sigma(F)$ of a POVM $F$ is the closed set 
$$\{x\in{X}:\,F(\Delta)\neq\mathbf{0},\,\forall\Delta\,\text{open},\,x\in\Delta\}.$$
\end{definition}
\noindent
\noindent
By the spectral theorem \cite{Reed}, there is a one-to-one correspondence between PV measures $E$ with spectrum in $\mathbb{R}$ and self-adjoint operators $B$,
the correspondence being given by
$$B=\int\lambda dE^B_{\lambda}.$$
\noindent
Notice that the spectrum of $E^B$ coincides with the spectrum of the corresponding self-adjoint operator $B$. Moreover, in this case a functional calculus can be developed. Indeed, if $f:{\mathbb R}\to{\mathbb R}$ is a measurable real-valued function, we can define the self-adjoint operator \cite{Reed}
$$f(B)=\int f(\lambda) dE^B_{\lambda}$$
\noindent
where, $E^B$ is the PVM corresponding to $B$. If $f$ is bounded, then $f(B)$ is bounded \cite{Reed}.
\noindent

\noindent
In the following we do not distinguish between a self-adjoint operator and the corresponding PVM and use the symbols $w-\lim$ and $u-\lim$ to denote the limit in the weak operator topology and the limit in the uniform operator topology respectively.
\begin{definition}
A POVM is regular if for any Borel set $\Delta$, 
$$w-\lim_{i\to\infty}F(G_j)=F(\Delta)=w-\lim_{i\to\infty}F(O_j)$$ 
\noindent
where, $\{G_j\}_{j\in\mathbb{N}}$, $\Delta\subset G_j$, is a decreasing family of open sets and $\{O_j\}_{j\in\mathbb{N}}$, $O_j\subset\Delta$, is a increasing family of compact sets and the convergence is in the weak operator topology.   
\end{definition}
\noindent
We recall \cite{Munkres} that  a topological space $(X,\tau)$ is second countable if it has a countable basis for its topology $\tau$; i.e., if there is a countable subset $\mathcal{B}$ of $\tau$ such that each member of $\tau$ is the union of members of $\mathcal{B}$. 
\begin{proposition}
A POVM defined on a Hausdorff locally compact, second countable space $X$ is regular. 
\end{proposition}
\begin{proof}
A  locally compact Hausdorff space is regular. (See Ref. \cite{Munkres}, page 205).  By the Urysohn's theorem, a second countable regular space is metrizable (see \cite{Munkres}, page 215). Moreover, the second countability implies the $\sigma$-compactness (\cite{Munkres}, page 289). In a metrizable $\sigma$-compact space the ring of Borel sets coincides with the ring of Baire sets (see page 25 in \cite{Berberian}) and the thesis comes from the fact that each Baire POVM is regular (see Theorem 18 in \cite{Berberian}). 
\end{proof}

\section{Uniform continuity of POVMs}
In the present section we single out conditions for a POVM to be uniformly continuous. We also analyze the concept of absolute continuity of POVMs and show the relationships with the concept of uniform continuity. At last we give a necessary condition for the norm-1 property of a uniformly continuous POVM. In the next section we give a necessary and sufficient condition for a commutative POVM to be uniformly continuous. 

\begin{definition}\cite{B2}
Let $F:\mathcal{B}(X)\to\mathcal{F(H)}$ be a POVM. 
$F$ is said to be uniformly continuous at $\Delta$ if, for any disjoint decomposition $\Delta=\cup_{i=1}^{\infty}\Delta_i$, 
$$\lim_{n\to\infty}\sum_{i=1}^{n}F(\Delta_i)=F(\Delta)$$
\noindent
 in the uniform operator topology. $F$ is said uniformly continuous if it is uniformly continuous at each $\Delta\in\mathcal{B}(X)$.
\end{definition}
We recall that a sequence of sets $\Delta_i$ is increasing if $\Delta_i\subset\Delta_{i+1}$. The limit operation $\lim_{i\to\infty}\Delta_i=\cup_{i=1}^\infty\Delta_i=\Delta$ is usually denoted by $\Delta_i\uparrow\Delta$. Analogously, $\Delta_i$ is a decreasing family of sets if $\Delta_{i+1}\subset\Delta_i$ and $\Delta_i\downarrow\Delta$ denotes the limit operation $\lim_{i\to\infty}\Delta_i=\cap_{i=1}^\infty\Delta_i=\Delta$.
\begin{proposition}\cite{B2}\label{up}
A POVM $F$ is uniformly continuous at $\Delta$ if and only if it is uniformly continuous from below at $\Delta$, i.e., for any increasing sequence $\Delta_i\uparrow\Delta$, 
\begin{equation*}
\lim_{n\to\infty}\| F(\Delta)-F(\Delta_i)\|=0.
\end{equation*}   
\noindent
$F$ is uniformly continuous if and only if it is uniformly continuous from below at each $\Delta$.
\end{proposition}

\begin{proposition}\cite{B2}\label{down}
$F$ is uniformly continuous if and only if, 
$\lim_{i\to\infty}\| F(\Delta_i)\|=0$ whenever $\Delta_i\downarrow\emptyset$.
\end{proposition}

Now, we introduce the concept of absolute continuity. 
\begin{definition}
Let $F:\mathcal{B}(X)\to\mathcal{F(H)}$ be a POVM and $\nu:X\to\mathbb{R}$ a regular measure. $F$ is absolutely continuous with respect to $\nu$ if there exists a number $c$ such that 
\begin{equation*}
\|F(\Delta)\|\leq c\,\nu(\Delta),\quad \forall\Delta\in\mathcal{B}(X).
\end{equation*}
\end{definition}
\begin{theorem}\cite{B2}\label{abs}
Let $F$ be absolutely continuous with respect to a regular finite measure $\nu$. Then, $F$ is uniformly continuous.
\end{theorem}
If $F$ is absolutely continuous with respect to an infinite measure, we have the following generalization of theorem \ref{abs}.
\begin{theorem}\cite{B2}\label{abc}
Suppose $F$ is absolutely continuous with respect to a regular measure $\nu$. Suppose, $\Delta$ is such that $\nu(\Delta)<\infty$. Then, $F$ is uniformly continuous at $\Delta$.
\end{theorem}

\subsection{Uniform continuity and norm-1 property}
\noindent
In the present subsection we give a necessary condition for the norm-1 property of uniformly continuous POVMs. First we recall the definition and the physical meaning of the norm-1 property.
\begin{definition}
A POVM $F:\mathcal{B}(X)\to\mathcal{F(H)}$ has the norm-1-property if $\|F(\Delta)\|=1$, for each $\Delta\in\mathcal{B}(X)$ such that $F(\Delta)\neq\mathbf{0}$. 
\end{definition}
\noindent
\noindent
The following proposition explains the physical meaning of the norm-1 property. 
\begin{proposition}\label{lim}
A POVM $F$ has the norm-1 property if and only if, for each $\Delta\in\mathcal{B}(X)$ such that $F(\Delta)\neq\mathbf{0}$, there is a sequence of unit vectors $\psi_n$ such that $\lim_{n\to\infty}\langle\psi_n, F(\Delta)\psi_n\rangle=1$.
\end{proposition}
If an observable is described by a PVM $E$ (sharp observable) then, for any Borel set $\Delta$ such that $E(\Delta)\neq\mathbf{0}$, there exists a unit vector $\psi$ for which $\langle\psi, E(\Delta)\psi\rangle=1$; i.e., the probability that a measure of the observable $E$ in the state $\psi$ gives a result in $\Delta$  is one. That is not true if an observable is described by a POVM $F$ (unsharp observable) since there are Borel sets $\Delta$ such that  $0<\langle\psi, F(\Delta)\psi\rangle<1$ for any vector $\psi$. We have here a relevant difference between sharp and unsharp observables. For example, suppose that  $E$ and $F$ are a sharp and an unsharp localization observable respectively ($E$ could refer to a non-relativistic particle and $F$ to the photon). Then, in the sharp case, for any set $\Delta$ there is a state $\psi$ such that the system is surely localizable in $\Delta$ (sharp localization) while in the unsharp case such a state does not exist in general (unsharp localization). That raises the problem of looking for conditions to be satisfied by the unsharp observables in order to ensure a kind of unsharp localization which is as close as possible to the sharp one. 
The norm-1 property is a possible answer to such a problem. Indeed, as a consequence of proposition \ref{lim}, if $F$ has the norm-1 property then, for any $\epsilon>0$, there is a pure state ${\psi}$ such that $\langle\psi, F(\Delta)\psi\rangle>1-\epsilon $.  In other words, the norm-1 property implies that, for any $\Delta$, there exists a preparation procedure such that the quantum mechanical system can be localized within $\Delta$ as accurately as desired although not sharply. 

The photon is an examples of a not sharply localizable system \cite{Newton,Wightman}. Recently \cite{B2} it was proved that the localization in phase space of massless relativistic particles does not satisfy the norm-1 property.  

In the case of uniformly continuous POVMs we have the following necessary condition for the norm-1 property.
\begin{theorem}\cite{B2}\label{norm1}
Let $F:X\to\mathcal{F(H)}$ be uniformly continuous and let $\sigma(F)$ be the spectrum of  $F$. Then, $F$ has the norm-1-property only if $\|F(\{x\})\|\neq 0$ for each $x\in\sigma(F)$.
\end{theorem}
\begin{theorem}\cite{B2}\label{norm2}
Let $F:\mathcal{B}(X)\to\mathcal{F(H)}$ be absolutely continuous with respect to a regular measure $\nu$. Then, $F$ has the norm-1 property only if $\|F(\{x\})\|\neq 0$ for each $x\in X$ such that $\nu(\{x\})<\infty$. 
\end{theorem}

\section{The commutative case}
In the present section we restrict ourselves to the case of commutative POVMs. In particular, we show that a commutative POVM is uniformly continuous if it admits a strong Feller Markov kernel. We start by recalling some results which  show the relationships between commutative POVMs and PVMs. In the following each POVM is assumed to be real, i.e., $F:\mathcal{B}(\mathbb{R})\to\mathcal{F}(H)$. Let $\Lambda$ be a subset of $\mathbb{R}$ and $\mathcal{B}(\Lambda)$ the corresponding Borel $\sigma$-algebra.
\begin{definition}
A real Markov kernel is a map $\mu:\Lambda\times\mathcal{B}(\mathbb{R})\to[0,1]$ such that,
\begin{itemize}
\item[1.] $\mu_{\Delta}(\cdot)$ is a measurable function for each $\Delta\in\mathcal{B}(\mathbb{R})$,
\item[2.] $\mu_{(\cdot)}(\lambda)$ is a probability measure for each $\lambda\in \Lambda$.
\end{itemize}
\end{definition}
\begin{definition}
Let $\nu$ be a measure on $\Lambda$. A map $\mu:\Lambda\times\mathcal{B}(\mathbb{R})\to[0,1]$ is a weak Markov kernel with respect to $\nu$ if:
\begin{itemize}
\item[1.] $\mu_{\Delta}(\cdot)$ is a measurable function for each $\Delta\in\mathcal{B}(\mathbb{R})$,
\item[2.] $0\leq\mu_{\mathbb{R}}(\lambda)\leq 1$,\quad $\nu-a.e.$,
\item[3.]$\mu_{\mathbb{R}}(\lambda)=1$, $\mu_{\emptyset}(\lambda)=0$,\quad $\nu-a.e.$,
\item[4.] for any sequence $\{\Delta_i\}_{i\in\mathbb{N}}$, $\Delta_i\cap\Delta_j=\emptyset$,  
$$\sum_i\mu_{(\Delta_i)}(\lambda)=\mu_{(\cup_i\Delta_i)}(\lambda),\quad \nu-a.e.$$
\end{itemize}
\end{definition}
\begin{definition}
The map $\mu:\Lambda\times\mathcal{B}(\mathbb{R})\to[0,1]$ is a weak Markov kernel with respect to a PVM $E:\mathcal{B}(\Lambda)\to\mathcal{E(H)}$ if it is a weak Markov kernel with respect to each measure $\nu_x(\cdot):=\langle E(\cdot)\,x,x\rangle$, $x\in\mathcal{H}$. 
\end{definition}
\noindent
In the following, by a weak Markov kernel $\mu$ we mean a weak Markov kernel with respect to a PVM $E$. Moreover the function $\lambda\mapsto\mu_{\Delta}(\lambda)$ will be denoted indifferently by $\mu_{\Delta}$ or $\mu_{\Delta}(\cdot)$.    
\begin{definition}
A POV measure $F:\mathcal{B}(\mathbb{R})\to\mathcal{F(H)}$ is said to be a smearing of a POV measure $E:\mathcal{B}(\Lambda)\to\mathcal{E(H)}$ if there exists a weak Markov kernel $\mu:\Lambda\times\mathcal{B}(\mathbb{R})\to[0,1]$ such that,
\begin{equation*}
F(\Delta)=\int_{\Lambda} \mu_{\Delta}(\lambda)d E_{\lambda}, \,\,\,\,\,\,\,\Delta\in\mathcal{B}(\mathbb{R}).
\end{equation*}
\end{definition}

\begin{example}
Let $Q$ be the position operator; i.e.,
\begin{align*}
Q:L^2(\mathbb{R})&\to L^2(\mathbb{R})\\
\psi(x)\in L^2(\mathbb{R})&\mapsto Q\psi:=x\psi(x)
\end{align*}
\noindent
A possible smearing of $Q$ is the optimal position POVM
\begin{align}\label{Q}
F^Q(\Delta)&=\frac{1}{l\,\sqrt{2\,\pi}}\int_{-\infty}^{\infty}\Big(\int_{\Delta}e^{-\frac{(x-y)^2}{2\,l^2}}\,d y\Big)\,dE^Q_x=\int_{-\infty}^{\infty}\mu_{\Delta}(x)\,dE^Q_x
\end{align}
\noindent
where, 
\begin{equation*}
\mu_{\Delta}(x)=\frac{1}{l\,\sqrt{2\,\pi}}\int_{\Delta}e^{-\frac{(x-y)^2}{2\,l^2}}\,d y 
\end{equation*}
\noindent
defines a Markov kernel and $E^Q$ is the spectral measure corresponding to the position operator $Q$.
\end{example}
\noindent
In the following, the symbol $\mu$ is used to denote both  Markov kernels and weak Markov kernels. The symbols $A$ and $B$ are used to denote self-adjoint operators. 
\begin{definition}
Whenever $F$, $A$, and $\mu$ are such that $F(\Delta)=\mu_{\Delta}(A)$, $\Delta\in\mathcal{B}(\mathbb{R})$, we say that $(F,A,\mu)$ is a von Neumann triplet. 
\end{definition}
 \begin{definition}
A Feller Markov kernel is a  Markov kernel $\mu_{(\cdot)}(\cdot):\Lambda\times\mathcal{B}(\mathbb{R})\to[0,1]$ such that the function 
$$G(\lambda)=\int_{\Lambda}f(t)\,d\mu_t(\lambda),\quad\lambda\in\Lambda$$
\noindent
is continuous and bounded whenever $f$ is continuous and bounded. 
\end{definition}
\noindent
The following theorem establishes a relationship between commutative POVMs and spectral measures and then provides a characterization of the former. Other characterizations of commutative POVMs and an analysis of the relationships between them can be found in Ref.s \cite{Ali,Holevo,Ali5,P1}.
\begin{theorem}\cite{B1}\label{Cha}
A real POVM $F:\mathcal{B}(\mathbb{R})\to\mathcal{F(H)}$ is commutative if and only if, there exists a bounded self-adjoint operator $A$ with spectrum $\sigma(A)\subset[0,1]$ and a Feller Markov Kernel 
$\mu:\Gamma\subset\sigma(A)\times\mathcal{B}(\mathbb{R})\to[0,1]$, $E(\Gamma)=\mathbf{1}$,
such that 
\begin{equation}\label{Smearing}
F(\Delta)=\int_{\Gamma}\mu_{\Delta}(\lambda)\,dE_{\lambda},\quad\Delta\in\mathcal{B}(\mathbb{R})
\end{equation}
\noindent
Moreover, $\mu$ separates the points in $\Gamma$ and $\mu_{\Delta}$ is continuous for each $\Delta$ in a ring $\mathcal{R}$ which generates the Borel $\sigma$-algebra of the reals.
\end{theorem}
\begin{corollary}\label{smearing}
 A POVM $F$ is commutative if and only if it is a smearing of a PVM $E$ with bounded spectrum.
\end{corollary}
We recall that the von Neumann algebra generated by a POVM $F$ is the von Neumann algebra generated by the set $\{F(\Delta),\,\Delta\in\mathcal{B}(\mathbb{R})\}$. 
\begin{definition}
If $A$ and $F$  in theorem \ref{Cha} generate the same von Neumann algebra then $A$ is named the sharp version of $F$.
\end{definition}
\begin{theorem}\label{unique}\cite{B4,B3}
The sharp version $A$ is unique up to almost everywhere bijections. 
\end{theorem}

\begin{definition}
A Markov kernel $\mu_{(\cdot)}(\cdot):[0,1]\times\mathcal{B}(\mathbb{R})\to[0,1]$ is said to be strong Feller if  $\mu_{\Delta}$ is a continuous function for each $\Delta\in\mathcal{B}(\mathbb{R})$.
\end{definition}
\begin{definition}
We say that a commutative POVM admits a strong Feller Markov kernel if there exists a strong Feller Markov kernel $\mu$ such that $F(\Delta)=\int \mu_{\Delta}(x)\,dE_{x}$, where $E$ is the sharp reconstruction of $F$.
\end{definition}
\begin{theorem}\cite{B1}\label{uni}
A commutative POVM $F:\mathcal{B}(\mathbb{R})\to\mathcal{F(H)}$ admits a strong Feller Markov kernel if and only if it is uniformly continuous.
\end{theorem}

\subsection{The meaning of uniformly continuous POVMs }
\noindent
In the present subsection we analyze the physical meaning of the smearing in the case of uniformly continuous POVMs. 

Let $E^Q(\Delta)$ be the spectral resolution of the position operator. We recall that $\langle\psi,E^Q(\Delta)\psi\rangle$ is interpreted as the probability that a perfectly accurate measurement (sharp measurement) of the position gives a result in $\Delta$.
\noindent
Then, a possible interpretation of equation (\ref{Q}) is that the POVM $F^Q$ is a randomization of $Q$. Indeed \cite{Prugovecki}, the outcomes of the measurement of the position of a particle depend on the measurement imprecision 
so that, if the sharp value of the outcome of the measurement of $Q$ is $x$ then the apparatus produces with probability $\mu_{\Delta}(x)$ a reading in $\Delta$. 

\noindent
It is worth noting that the Markov kernel 
\begin{equation*}
\mu_{\Delta}(x)=\frac{1}{l\,\sqrt{2\,\pi}}\int_{\Delta}e^{-\frac{(x-y)^2}{2\,l^2}}\,d y 
\end{equation*}
\noindent
 in equation (\ref{Q}) above is such that the function $x\mapsto\mu_{\Delta}(x)$ is continuous for each interval $\Delta\in\mathcal{B}(\mathbb{R})$. (See subsection \ref{position}) The continuity of $\mu_{\Delta}$ means that if two sharp values $x$ and $x'$ are very close to each other then, the corresponding random diffusions are very similar, i.e., the probability to get a result in $\Delta$ if the sharp value is $x$ is very close to the probability to get a result in $\Delta$ if the sharp value is $x'$. That is quite common in important  physical applications and  seems to be reasonable from the physical viewpoint. It is then natural to look for general conditions which ensure the continuity of $x\mapsto\mu_{\Delta}$. It can be proved \cite{B1} that, in general, the continuity does not hold for all the Borel sets $\Delta$ but only for a ring of subsets which generates the Borel $\sigma$-algebra of the reals. (Anyway, that is sufficient to prove the weak convergence of $\mu_{(\cdot)}(x)$ to $\mu_{(\cdot)}(x')$.) 

That $x\mapsto\mu_{\Delta}$ cannot  in general be continuous for each Borel set $\Delta$ is already established by theorem \ref{uni} which states that the map $x\mapsto\mu_{\Delta}$ is continuous for each $\Delta\in\mathcal{B}(\mathbb{R})$ if and only if the POVM $F$ is uniformly continuous.   Therefore, we can say that theorem \ref{uni} suggests an interpretation of the uniformly continuous POVMs; i.e, $F$ is uniformly continuous if and only if the smearing $\mu_{\Delta}(A)=F(\Delta)$ can be realized by a Markov kernel $\mu$ with the property that, for any $\Delta\in\mathcal{B}(\mathbb{R})$, the probability to get a result in $\Delta$ if the sharp value is $x$ is very close to the probability to get a result in $\Delta$ if the sharp value is $x'$. 



\section{Analysis of some relevant physical examples}

\subsection{Bounded position operator}
Let us consider the unsharp position operator defined as follows.
\begin{align}
Q^f(\Delta)&:=\int_{[0,1]}\mu_{\Delta}(x)\,dQ_x,\quad\Delta\in\mathcal{B}(\mathbb{R}),\\
\mu_{\Delta}(x)&:=\int_{\mathbb{R}}\chi_{\Delta}(x-y)\, f(y)\,dy,\quad x\in[0,1]\notag
\end{align}
\noindent
where, $f$ is a positive, bounded, continuous function such that $f(x)=0$, $x\notin [0,1]$, 

$$\int_{[0,1]} f(x)dx=1,$$ 
\noindent
and $Q_x$ is the spectral measure corresponding to the position operator
\begin{align*}
Q:L^2([0,1])&\to L^2([0,1])\\
\psi(x)&\mapsto Q\psi:=x\psi(x)
\end{align*}
\noindent
$Q^f$ is absolutely continuous with respect to the measure 
$$\nu(\Delta)=M\int_{\Delta\cap[-1,1]}dx.$$ 
\noindent
Indeed, for each $\psi\in\mathcal{H}$, $|\psi|^2=1$,
\begin{align*}
\langle\psi,Q^f(\Delta)\psi\rangle=\int_{[0,1]}\mu_{\Delta}(x)\,\psi^2(x)\,dx\leq M\int_{\Delta\cap[-1,1]}dx
\end{align*}
\noindent
where, the inequality
$$\mu_{\Delta}(x)=\int_{\Delta} f(x-y)\,dy\leq M\int_{\Delta\cap[-1,1]} dx$$
\noindent
has been used.

\noindent
Therefore, by theorem \ref{abs}, $Q^f(\Delta)$ is uniformly continuous. Moreover, the continuity of $f$ assures the continuity of $\mu_{\Delta}$ for each $\Delta\in\mathcal{B}(\mathbb{R})$ so that $\mu$ is a Feller Markov kernel. At last, the norm-1 property of $Q^f$ is forbidden by theorem \ref{norm1}.

\subsection{Phase observable}
Following Ref. \cite{He}, we use the following representation of the phase observable 
$$E:\mathcal{B}[0,2\pi)\to\mathcal{F(H)},$$
\begin{equation}
E(\Delta)=\sum_{m,n=1}^\infty\langle\psi_n\vert\psi_m\rangle\frac{1}{2\pi}\int_{\Delta}e^{i(n-m)x}\,dx\,\vert n\rangle\langle m\vert
\end{equation}
\noindent
where $\vert\psi_n\rangle$ is a sequence of unit vectors in $\mathcal{H}$.

The POVM $E$ just defined is covariant with respect to the phase shift operator; i.e.,
$$e^{iN\theta}E(\Delta)e^{-iN\theta}=E(\Delta\oplus\theta)$$
\noindent
where, the symbol $\oplus$ denotes addition modulo $2\pi$.

1) First, we analyze the case $\langle\psi_n\vert\psi_m\rangle=\delta_{n,m}$, $n,m\neq s,t$, $s\neq t$, $|\langle\psi_s\vert\psi_t\rangle|<1$. We have
$$E_1(\Delta)=\frac{1}{2\pi}|\Delta|\,\mathbf{1}+\frac{1}{2\pi}\langle\psi_s\vert\psi_t\rangle\int_{\Delta}e^{i(s-t)x}\,dx\,\vert s\rangle\langle t\vert+\frac{1}{2\pi}\langle\psi_t\vert\psi_s\rangle\int_{\Delta}e^{i(t-s)x}\,dx\,\vert t\rangle\langle s\vert$$
\noindent
where, $|\Delta|$ is the Lebesgue measure of $\Delta$. We notice that $E_1$ is absolutely continuous with respect to the Lebesgue measure on $[0,2\pi)$. Indeed,  
\begin{align}
\|E_1(\Delta)\|\leq\frac{1}{2\pi}|\Delta|+2\frac{1}{2\pi}|\Delta|\,\|(\vert s\rangle\langle t\vert)\|\leq \frac{3}{2\pi}|\Delta|.
\end{align}
\noindent
Therefore, by theorem \ref{abs}, $E_1$ is uniformly continuous and, by theorem \ref{norm1}, it cannot have the norm-1 property since $|\{x\}|=0$ for each $x\in[0,2\pi)$. 

2) If $\psi_n=\psi$, $\forall n\in\mathbb{N}$, we have the canonical phase observable 
$$E_{can}(\Delta)=\frac{1}{2\pi}|\Delta|\,\mathbf{1}+\frac{1}{2\pi}\sum_{n\neq m}\int_{\Delta}e^{i(n-m)x}\,dx\,\vert n\rangle\langle m\vert.$$

\noindent
In \cite{He} it is proved that $E_{can}(\Delta)$ has the norm-1 property. Moreover, we notice that $E_{can}(\{x\})=\mathbf{0}$ for each $x\in X$. Therefore, by theorem \ref{norm1}, $E_{can}(\Delta)$ cannot be uniformly continuous.

Finally, we remark that the phase space observables we analyzed in items 1) and 2) are not commutative.

\subsection{Unsharp number observable}
The unsharp number observable describes a photon detector with efficiency $\epsilon$ less than 1, and is represented by the commutative POVM 
\begin{equation}\label{number}
F_n^{\epsilon}:=\sum_{m=n}^{\infty}
\begin{pmatrix}
m\\
n
\end{pmatrix}
\epsilon^n (1-\epsilon)^{m-n}\vert m\rangle\langle m\vert.
\end{equation}
Notice that \cite{B7} $F_n^{\epsilon}$ is an unsharp version of the number operator $N$. That can be seen by introducing the functions
\begin{equation}\label{Markov}\mu_n(m)=
\begin{cases}
0 & if \quad n>m\\
\begin{pmatrix}
m\\
n
\end{pmatrix}
\epsilon^n (1-\epsilon)^{m-n} & if \quad n\leq m
\end{cases}
\end{equation}
\noindent
which are such that

$$F_n^{\epsilon}=\mu_n(N),$$
\begin{equation}\label{mu}
\sum_{n=0}^{\infty}\mu_n(m)=1,\quad\forall m\in\mathbb{N},
\end{equation}
\noindent
\noindent

\noindent
For each $n$, $F_n^\epsilon$ is compact.\footnote{Since \begin{equation*}
\begin{split}
\lim_{m\to\infty}
\begin{pmatrix}
m\\
n
\end{pmatrix}
\epsilon^n (1-\epsilon)^{m-n}=\lim_{m\to\infty}\frac{\epsilon^n}{n!}\, [m(1-\epsilon)^{m-n/n} ]\cdots[(m-n)(1-\epsilon)^{m-n/n} ]=0,
\end{split}
\end{equation*}
$F_n^{\epsilon}$ is compact for every $n\in\mathbb{N}$ (see pages 234-235 in \cite{Nagy}).} Therefore, the unsharp number observable is not uniformly continuous. Indeed, being the space of compact operators closed with respect to the uniform operator topology, the uniform continuity would imply that $\mathbf{1}=u-\lim_{n\to\infty}\sum_{i=1}^{n}F_i^\epsilon$ which is not possible since $\mathbf{1}$ is not compact.
\noindent
Moreover, $F_n^\epsilon$ has not the norm-1 property. 
 Indeed, if $\|F_n^\epsilon\|=1$ then (see theorem VI.6 in \cite{Reed}) $\sup_{\lambda\in\sigma(F_n^\epsilon)}|\lambda|=1$ and, by the compactness of the spectrum, $1\in\sigma(F_n^\epsilon)$. On the other hand, the compactness of $F_n^\epsilon$ implies that if $\|F_n^\epsilon\|=1$ then, $1$ is an eigenvalue of $F_n^\epsilon$ but that contradicts equation (\ref{number}) where the eigenvalues of $F_n^\epsilon$ are  $0$ and

\begin{equation*}
\begin{pmatrix}
m\\
n
\end{pmatrix}
\epsilon^n (1-\epsilon)^{m-n},\quad m\geq n,
\end{equation*}
which, by equation (\ref{mu}), are strictly less than $1$. The only effect with eigenvalue $1$ is $F_0^{\epsilon}$.

\subsection{Position and Momentum}\label{position}

\noindent
In the following $\mathcal{H}=L^2(\mathbb{R})$, $Q$ and $P$ denote the sharp position and momentum observables respectively while $\ast$ denotes convolution, i.e. $(f\ast g)(x)=\int f(y)g(x-y)d y$. 

\noindent
Let us consider the joint position-momentum POVM \cite{Ali,Busch,Davies,Holevo,Prugovecki,Schroeck} 
\begin{equation*} \label{phase}
F(\Delta\times\Delta')=\int_{\Delta\times\Delta'}U_{q,p}\,\gamma\,U^*_{q,p}\,d q\, d p
\end{equation*}
where, $U_{q,p}=e^{-iqP}e^{ipQ}$ and $\gamma=\vert f\rangle\langle f\vert$, $f\in L^2(\mathbb{R})$, $\|f\|_2=1$.
The marginal
\begin{equation}
\label{approximate}
Q^{f}(\Delta):=F(\Delta\times\mathbb{R})=\int_{-\infty}^{\infty}({\bf 1}_{\Delta}\ast\vert f\vert^2)(x)\,dQ_x,\quad\Delta\in\mathcal{B}(\mathbb{R}),
\end{equation}
is an unsharp position observable while the map $\mu_{\Delta}(x):={\bf 1}_{\Delta}\ast \vert f(x)\vert^2$ defines a Markov kernel. 

\noindent
First we notice that $Q^f$ is absolutely continuous with respect to the Lebesgue measure. Indeed,
\begin{align*} 
Q^f(\Delta)=F(\Delta\times\mathbb{R})&=\int_{\Delta\times\mathbb{R}}U_{q,p}\,\gamma\,U^*_{q,p}\,d q\, d p\\
&=\int_{\Delta}\,dq\int_{\mathbb{R}} U_{q,p}\,\gamma\,U^*_{q,p}\,d p\\
&=\int_{\Delta}\widehat{Q}(q)\,dq\leq\int_{\Delta}\mathbf{1}\,dq
\end{align*}

\noindent
where, 
$$\widehat{Q}(q)=\int_{\mathbb{R}}U_{q,p}\,\gamma\,U^*_{q,p}\,dp.$$
\noindent
Then, by theorem \ref{abc}, $Q^f$ is uniformly continuous at each Borel set $\Delta$ with finite Lebesgue measure and, by theorem \ref{norm2}, $Q^f$ cannot have the norm-1 property since $\|Q^f(\{x\})\|\leq|\{x\}|=0$.

Moreover, it is worth remarking that, the uniform continuity does not hold in general and that there are sets for which the norm-1 property is satisfied. That can be shown by working out the details of the following particular relevant case. Let us set

$$f^2(x)=\frac{1}{l\,\sqrt{2\,\pi}}\,e^{(-\frac{x^2}{2\,l^2})},\quad l\in\mathbb{R}-\{0\},$$ 
\noindent
in (\ref{approximate}). The corresponding unsharp position POVM is

\begin{align*}
Q^{f}(\Delta)&=\int_{-\infty}^{\infty}\Big(\int_{\Delta}\vert f(x-y)\vert^2)\,d y\Big)\,dQ_x\\
&=\frac{1}{l\,\sqrt{2\,\pi}}\int_{-\infty}^{\infty}\Big(\int_{\Delta}e^{-\frac{(x-y)^2}{2\,l^2}}\,d y\Big)\,dQ_x=\int_{-\infty}^{\infty}\mu_{\Delta}(x)\,dQ_x
\end{align*}
\noindent
where, 
\begin{equation}\label{PM}
\mu_{\Delta}(x)=\frac{1}{l\,\sqrt{2\,\pi}}\int_{\Delta}e^{-\frac{(x-y)^2}{2\,l^2}}\,d y 
\end{equation}
\noindent
defines a Markov kernel.

\noindent
Now, we consider the family of sets $\Delta_i=(-\infty, a_i)$, $\lim_{i\to \infty}a_i=-\infty$ such that $\Delta_i\downarrow\emptyset$, and prove that $\|Q^f(\Delta_i)\|=1$, $\forall i\in\mathbb{N}$. For each  $i\in\mathbb{N}$,
\begin{align*}
\lim_{x\to-\infty}\mu_{\Delta_i}(x)&=\lim_{x\to-\infty}\frac{1}{l\,\sqrt{2\,\pi}}\int_{\Delta_i}e^{-\frac{(x-y)^2}{2\,l^2}}\,d y\\
&=\lim_{x\to-\infty}\frac{1}{l\,\sqrt{2\,\pi}}\int_{(-\infty,\, a_i-x)}e^{-\frac{y^2}{2\,l^2}}\,d y=\frac{1}{l\,\sqrt{2\,\pi}}\int_{-\infty}^{\infty}e^{-\frac{y^2}{2\,l^2}}\,d y=1.  
\end{align*}
\noindent
Therefore, if 
$$\psi_n=\chi_{[-n,-n+1]}(x),$$
\begin{align}\label{n}
\lim_{n\to\infty}\langle\psi_n,Q^f(\Delta_i)\psi_n\rangle&=\lim_{n\to\infty}\int_{-\infty}^{\infty}\mu_{\Delta_i}(x)\vert\psi_n(x)\vert^2\,dx\\\nonumber
&=\lim_{n\to\infty}\int_{[-n,-n+1]}\mu_{\Delta_i}(x)\,dx=1.
\end{align} 
\noindent
so that,  $\|Q^f(\Delta_i)\|=1$, $\forall i\in\mathbb{N}$, $\lim_{i\to\infty}\|Q^f(\Delta_i)\|=1$ and, by proposition \ref{down}, $Q^f$ cannot be uniformly continuous. 

A second consequence of equation (\ref{n}) is that $Q^f$ has the norm-1 property at each set $(-\infty, a)$, $a\in\mathbb{R}$. 

At last, it is worth noticing \cite{B1} that $\mu_{\Delta}$ is continuous for each interval $\Delta$. Indeed,  
\begin{align*}
\vert\mu_{\Delta}(x)-\mu_{\Delta}(x')\vert&=\frac{1}{l\,\sqrt{2\,\pi}}\Big\vert \int_{\Delta}e^{-\frac{(x-y)^2}{2\,l^2}}\,dy-\int_{\Delta}e^{-\frac{(x'-y)^2}{2\,l^2}}\,dy \Big\vert\\
&=\frac{1}{l\,\sqrt{2\,\pi}}\Big\vert\int_{\Delta_x}e^{-\frac{(y)^2}{2\,l^2}}-\int_{\Delta_{x'}} e^{-\frac{(y)^2}{2\,l^2}}\,dy\Big\vert\leq\frac{1}{l\,\sqrt{2\,\pi}}\Big\vert\int_{\overline{\Delta}} e^{-\frac{(y)^2}{2\,l^2}}\,dy\Big\vert
\end{align*}
\noindent
where, 
$$\Delta_x=\{z\in\mathbb{R}\,\vert\,z=y-x,\,y\in\Delta\},\quad\Delta_{x'}=\{z\in\mathbb{R}\,\vert\,z=y-x',\,y\in\Delta\}$$ 
\noindent
and,
$$\overline{\Delta}=(\Delta_x-\Delta_{x'})\cup(\Delta_{x'}-\Delta_x).$$

\noindent
Therefore, $\vert x-x'\vert\leq\epsilon$ implies,
\begin{equation*}
\vert\mu_{\Delta}(x)-\mu_{\Delta}(x')\vert\leq\frac{1}{l\,\sqrt{2\,\pi}}\Big\vert\int_{\overline{\Delta}} e^{-\frac{(y)^2}{2\,l^2}}\,dy\Big\vert\leq\frac{1}{l\,\sqrt{2\,\pi}}\,\int_{\overline{\Delta}}\,dy=\frac{\sqrt{2}}{l\,\sqrt{\pi}}\,\epsilon.
\end{equation*}

\end{document}